\documentclass[10pt,
aps,times,prl,
superscriptaddress,reprint]{revtex4-1}

\usepackage{graphicx}
\usepackage{dcolumn}
\usepackage{bm}
\usepackage{lmodern}
\usepackage{amsmath,amssymb,amsthm,amsfonts, mathtools}
\usepackage{dsfont}
\usepackage{mathrsfs}
\usepackage{mathtools}
\usepackage{mathbbol}
\usepackage{hyperref}
\usepackage{xcolor}

\newcommand{\ssigma}{\boldsymbol{\sigma}}
\newcommand{\rrho}{\hat{\boldsymbol{\rho}}}

\newcommand{\Tr}{\mathrm{Tr}}
\newcommand{\Id}{\mathds{I}}
\newcommand{\SL}{\mathrm{SL}}

\newcommand{\diag}{\mathrm{diag}}

\newcommand{\reals}{\mathds{R}}
\newcommand{\complex}{\mathds{C}}

\theoremstyle{plain}
\newtheorem{thm}[subsubsection]{Theorem}

\theoremstyle{definition}
\newtheorem{defn}[subsubsection]{Definition}

\newtheorem{propos}[subsubsection]{Proposition}

\begin{document}
\title{Nonclassical steering with two-mode Gaussian states}
\author{Massimo Frigerio}
\email{Electronic address: m.frigerio49@campus.unimib.it}
\affiliation{Dipartimento di Fisica dell'Universit\`a degli Studi di Milano-Bicocca, 
I-20126 Milano, Italy} 
\author{Claudio Destri}
\email{Electronic address: claudio.destri@unimib.it}
\affiliation{Dipartimento di Fisica dell'Universit\`a degli Studi di Milano-Bicocca, 
I-20126 Milano, Italy}  
\author{Stefano Olivares}
\email{Electronic address: stefano.olivares@fisica.unimi.it}
\affiliation{Dipartimento di Fisica ``Aldo Pontremoli''  dell'Universit\`a degli 
Studi di Milano, I-20133 Milano, Italy}
\author{Matteo G. A. Paris}
\email{Electronic address: matteo.paris@fisica.unimi.it}
\affiliation{Dipartimento di Fisica  ``Aldo Pontremoli'' dell'Universit\`a degli 
Studi di Milano, I-20133 Milano, Italy}
\date{\today}
\begin{abstract}
Singularity or negativity of Glauber P-function is a widespread notion of nonclassicality, with
important implications in quantum optics and with the character of an irreducible resource. Here
we explore how P-nonclassicality may be generated by conditional Gaussian measurements on bipartite
Gaussian states. This \emph{nonclassical steering} may occur in a weak form, which does not imply
entanglement, and in a strong form that implies EPR-steerability and thus entanglement. We show
that field quadratures are the best measurements to remotely generate nonclassicality, and exploit
this result to derive necessary and sufficient conditions for weak and strong nonclassical steering.
For two-mode squeezed thermal states (TMST), weak and strong nonclassical steering coincide, and
merge with the notion of EPR steering. This also provides a new operational interpretation for
P-function nonclassicality as the distinctive feature that allows one-party entanglement verification
on TMSTs.
\end{abstract}
\maketitle
The classification of quantum correlations is a very active front of research 
since the early days of quantum mechanics. In this Letter, we investigate 
\emph{quantum steering}, a class of asymmetric quantum correlations stronger than entanglement \cite{wern89}, but weaker  than violation of Bell's inequality \cite{UCNG20,quint15}, that was introduced in relation to the EPR argument \cite{schr35,EPR}, to indicate the possibility of one party
to collapse (or \emph{steer}) the wavefunction of the other party into 
different quantum states by means of suitable measurements. Despite this 
early appearance, steering received firm mathematical bases only recently \cite{wise07,jone07}, and we refer 
to this definition as \emph{EPR steering}, particularly in the context of continuous-variable (CV) systems \cite{XXMTHA17}. The central idea of EPR steering is 
to use the influence of the measurements performed by one party (say Alice) to convince the other party (say Bob) that the shared state was entangled: 
if the initial correlated state allows for such a task, it is called \emph{EPR-steerable} by Alice. Steering is now widely considered a fundamental resource for quantum communication tasks \cite{bran12,heAd15,gom15,woll16,xia17,deng17} and many criteria for its detection have been explored \cite{kog07,schn13,cwlee13,ji16}.
\par
Independently of quantum correlations, a variety of other
concepts of nonclassicality have been put forward \cite{ferrpar12}. For 
CV quantum systems, the nonclassicality of a quantum state $\rrho$ 
is often characterized using the singularity of its Glauber P-function \cite{cahillglauber69,glauber69,lee95},
i.e. its expansion  onto coherent states $\vert \alpha \rangle$ 
($\alpha \in \complex$) according to:
 $\rrho  =  \int_{\complex} \mathrm{d}^{2}  
\alpha  P \left[ \rrho \right] \left( \alpha \right)  \vert \alpha 
\rangle \langle \alpha \vert$. 
The main reason for the wide use of the P-function is that it leads to the most \emph{physically inspired} notion of nonclassicality. It has direct empirical consequences, for example in quantum optics, where it is known to be necessary for antibunching and sub-Poissonian photon statistics \cite{mand86}. Viceversa, 
classicality  according to the P-function implies the empirical adequacy 
of Maxwell's Equations for the phenomenological description of the 
corresponding state of light. Moreover, P-nonclassical states
are usually harder to fabricate \cite{brau05,alb16}, 
thereby giving a \emph{resource} character to this type of nonclassicality \cite{yadbind18,kwtanvolk19}. In this paper, we investigate the possibility of 
\emph{steering} nonclassicality with two-mode Gaussian states, i.e. 
manipulating and generating it remotely \cite{alb17}, and 
introduce the concepts of {\em weak} and {\em strong} nonclassical 
steering for bipartite Gaussian states.
\par
As a first step, let us briefly review the definition of P-nonclassicality. 
The P-function is a member of a continuous family of phase
space quasiprobability distributions, known as $s$-ordered Wigner functions 
and defined according to \cite{cahillglauber69}:
\begin{equation}
\label{eq:wignerS}
    W_{s} \left[ \rrho  \right] \left( X \right) \ = \   \int_{\reals^{n}} \frac{d^{2n} \Lambda}{(2\pi^2)^{n}} \ \mathrm{e}^{ \frac{1}{4} s \vert \Lambda \vert^{2}  +  i \Lambda^{T} \boldsymbol{\Omega} X } \ \chi \left[ \rrho \right] \left( \Lambda \right)
\end{equation}
for $s \in [ - 1 , 1]$. Here the characteristic function \cite{FOPGauss} is defined as $\chi \left[ \rrho \right] ( \Lambda) = \Tr [  \rrho\, \mathrm{e}^{i \Lambda^{T} \hat{\mathbf{R}}} ]$, where $\hat{\mathbf{R}} = (\hat{x}_{1}, \hat{p}_{1},..., \hat{x}_{n}, \hat{p}_{n})^{T}$ is the vector of the canonical operators (or quadrature operators), related to the mode operators by $\hat{x}_j = (\hat{a}_j + \hat{a}^\dagger_j)/\sqrt{2}$, $\hat{p}_j = -i(\hat{a}_j - \hat{a}_j^{\dagger})/\sqrt{2}$. The case $s= 1$ corresponds precisely to the Glauber P-function, which is therefore the most singular of the family and can behave even more singularly than a tempered distribution. When the P-function of a CV quantum state $\rrho$ is not positive 
semidefinite \cite{daman18} and/or it is more singular than a delta
 distribution, the state 
is termed \emph{nonclassical} \cite{mand86,lee91,lutk95}. The so-called \emph{nonclassical depth} of a CV state $\rrho$ is then the quantity $\mathfrak{t} = \frac12 (1- s_{m})  $, where $s_{m}$ is the largest real number such that $W_{s} \left[ \rrho \right] (X)$ is nonsingular $\forall s < s_{m}$. Thus $\rrho$ is nonclassical if $\mathfrak{t} > 0$ and classical if $\mathfrak{t}=0$.
\par 
Let us now consider a Gaussian state $\rrho_{AB}$ of mode $A$ 
controlled by Alice, and mode $B$ controlled by Bob. We write its 
characteristic function as \cite{FOPGauss,olivares}:
\begin{equation}
\label{eq:wignergauss}
\chi  \left[ \rrho_{AB} \right] \left( \Lambda \right) \ = \ \exp \left\{  - \frac{1}{2} \Lambda^{T} \boldsymbol{\sigma} \Lambda \  - \ i \Lambda^{T}   \langle \hat{\mathbf{R}} \rangle   \right\}
\end{equation}
where the covariance matrix (CM) reads $\ssigma_{jk} =  \frac{1}{2} \langle  \{ \hat{R}_{j} , \hat{R}_{k} \} \rangle - \langle \hat{R}_{j} \rangle \langle \hat{R}_{k} \rangle $, with $\langle \hat{\mathbf{R}} \rangle = \Tr_{AB} [ \rrho_{AB} \hat{\mathbf{R}} ]$. The uncertainty relations (UR) may be recast into a constraint 
on the CM associated with physical states \cite{sera07}, i.e.
$\ssigma + i\,  \boldsymbol{\Omega}/2 \geq 0$, 
where $\boldsymbol{\Omega} = \oplus_{j=1}^{n} \boldsymbol{\omega}$ (for $n$ modes) and $\boldsymbol{\omega} =  i \sigma_{y}$ is the standard symplectic form \footnote{$\sigma_{y}$ being the second Pauli matrix.}.
\par
Since $ \chi [ \rrho ] ( \Lambda )$ is a Gaussian function
on phase space whenever $\rrho$ is a Gaussian state, it is straightforward to conclude from Eq.(\ref{eq:wignergauss}) 
and Eq.(\ref{eq:wignerS}) that $\rrho$ is nonclassical if and only if the least eigenvalue of $\ssigma$ is smaller than $\frac{1}{2}$. Examples of classical 
Gaussian states are coherent and thermal states, while squeezed vacuum states 
are always nonclassical. In the following, we will be interested in characterizing 
how quantum correlations in the joint Gaussian quantum state $\rrho_{AB}$ may be exploited to influence the nonclassicality of one mode (say $A$) by Gaussian measurements on the other one (mode $B$). In doing so, Local Gaussian Unitary Transformations (LGUTs) do not affect these correlations, and therefore we 
may freely perform LGUTs on the two modes to bring $\rrho_{AB}$ into a 
simpler form. In particular, by means of LGUTs a two-mode Gaussian state 
can always be brought into the so-called \emph{canonical form} \cite{sera04, DGCZ00, olivares}, for which 
the CM $\ssigma$ can be decomposed in $2 \times 2$ diagonal blocks $\ssigma \ = \ \left(  \begin{array}{cc} \mathbf{A} & \mathbf{C} \\ \mathbf{C}^{T} &  \mathbf{B}      \end{array}  \right)$
with $\mathbf{A} = a \cdot \Id_{2}$, $\mathbf{B} = b \cdot \Id_{2}$ and $\mathbf{C} = \diag ( c_{1}, c_{2} )$, while $a,b,c_{1}, c_{2} \in \reals$. We now note that the \emph{unconditional state} of mode $A$, defined either as the state that Alice uses to describe her mode without knowing anything about Bob's mode \emph{or} as the state she assigns to her mode by assuming that Bob has performed some measurement on his mode without letting her know the outcome, is given by $\rrho_{A} = \Tr_{B} \left[ \rrho_{AB} \right]$ and has a CM $\ssigma_{A} = \mathrm{A}$. Since 
the UR imply that $a \geq \frac{1}{2}$ this means that 
$\rrho_{A}$ must be classical. The same holds true for mode $B$, thus we may 
say that given a two-mode Gaussian state $\rrho_{AB}$ in canonical form, 
neither of the two modes has any intrinsic nonclassicality. 
Based on this observation, we advance the following definition:
\begin{defn}
\label{def:wnscanon}
A two-mode Gaussian state $\rrho_{AB}$ in canonical form is called \emph{weakly nonclassically steerable} (WNS) from mode $B$ to mode $A$ ($B \rightarrow A$) 
if there exists a Gaussian positive operator-valued measure (POVM) $ \{ \hat{\boldsymbol{\Pi}_{\alpha} } \}_{\alpha \in \complex }$ on mode $B$ 
such that the \emph{conditional state of mode $A$} after such measurement and communication of the outcome $\alpha$:
\begin{equation}
    \rrho_{c, \alpha } \ = \ \frac{1}{p_{\alpha}} \Tr_{B} \left[ \rrho_{AB} \left( \Id_{A} \otimes \hat{\boldsymbol{\Pi}}_{\alpha} \right) \right]
\end{equation}
is \emph{nonclassical}, where $p_{\alpha} = \Tr_{AB} [ \rrho_{AB} ( \Id_{A} \otimes \hat{\boldsymbol{\Pi}}_{\alpha} ) ]$ is the probability of observing the outcome $\alpha \in \complex$. 
\end{defn}
Let us now deduce a simple criterion to discern weakly nonclassically steerable states, starting with the following proposition:
\begin{propos}
The least classical (i.e. with highest possible nonclassical depth) conditional 
state $\rrho_{c, \alpha}$ of mode $A$ attainable with Gaussian measurements on 
mode $B$ of a two-mode Gaussian state $\rrho_{AB}$ in canonical form is reached
by quadrature detection on mode $B$, either of the $\hat{x}_{B}$ quadrature if 
$\vert c_{2} \vert \geq \vert c_{1} \vert$, or of the $\hat{p}_{B}$ quadrature 
otherwise. 
\end{propos}
\begin{proof}
Let us denote by $\ssigma_{c}$ the CM of the conditional state: one can show 
that it does not depend on the outcome $\alpha$, but just on the CM of the 
POVM performed on $B$. Therefore, $\rrho_{AB}$ in canonical form is WNS if 
and only if there exists a Gaussian POVM such that the least eigenvalue of 
$\ssigma_{c}$ is smaller than $\frac{1}{2}$. The effects of the most general 
Gaussian POVM on a single mode may be written as $\hat{\boldsymbol{\Pi}}_{\alpha} = D(\alpha) \rrho_{G}D^{\dagger}( \alpha)/\pi$ where $D(\alpha) = \exp \{ \alpha \hat{a} - \alpha^{*} \hat{a}^{\dagger} \}$ is the displacement operator and $\rrho_{G}$ 
is a single-mode Gaussian state with
$\langle \hat{\mathbf{R}} \rangle = 0$.
Furthermore, we may choose the following convenient parametrization for the CM $\ssigma_{M}$ of $\rrho_{G}$:
\begin{equation} 
\label{eq:CMmeas} \ssigma_{M} \ = \ \frac{1}{2  \mu \mu_{s} } \left( 
\begin{array}{cc} 1 + \kappa_s \cos \phi & - \kappa_s\sin \phi \\ - \kappa_s \sin \phi &  1 -\kappa_s \cos \phi \end{array} \right)
\end{equation}
where $\mu = \Tr[\rrho_{G}^2] \in [0, 1]$ is the purity of $\rrho_{G}$, $\mu_{s} = [1 + 2 \sinh^{2} r_{m} ]^{-1}$, $\kappa_s =\sqrt{1 - {\mu_{s}}^{2}}$, $r_{m}$ being the squeezing parameter of the state,
and $\phi \in [ 0, 2 \pi )$ is a phase. According to a well-known result \cite{olivares,ESP02,GC02}, the conditional CM is then given by the Schur complement \cite{matrx} of $\ssigma$ with respect to $( \mathbf{B} + \ssigma_{M})$, explicitly $ \ssigma_{c} = \mathbf{A} - \mathbf{C}^{T} \left( \mathbf{B} + \ssigma_{M} \right)^{-1} \mathbf{C} $.
Since $\mathbf{A}$ is diagonal, the minimum $\lambda_{m}$ 
(over all possible CMs $\ssigma_{M}$) 
of the smallest eigenvalue of $\ssigma_{c}$ is attained for the supremum of the 
greatest eigenvalue of $\mathbf{C}^{T} ( \mathbf{B} + \ssigma_{M} )^{-1} 
\mathbf{C}$, which is positive semidefinite. By explicit calculation, this 
supremum requires $\phi = 0$ if $\vert c_{2} \vert \geq \vert c_{1} \vert$, 
and $ \phi = \pi$ otherwise. The resulting expression is a monotonic 
decreasing function of $\mu_{s}$, since one can see by inspection that its first derivative with respect to $\mu_{s}$ is always nonpositive. Therefore, one 
needs to set $\mu_{s} = 0$ in order to attain the supremum and in this limit 
the value of $\mu$ becomes irrelevant. The limit $\mu_{s} \to 0$ makes the 
Gaussian POVM $\hat{\boldsymbol{\Pi}}_{\alpha}$ to collapse into the 
spectral measure of the $\hat{x}$($\hat{p}$) quadrature for $\phi = 0(\pi)$. 
\end{proof}
This result immediately leads us to the aforementioned criterion:
\begin{propos}
\label{propos:WNScanon}
A two-mode Gaussian state $\rrho_{AB}$ in canonical form is WNS ($B \rightarrow A$) if and only if the parameters of its CM satisfy:
\begin{equation}
\label{eq:WNScanon}
    a -  c^{2}/b   <  1/2\,, \quad c = \max \{ \vert c_{1} \vert, \vert c_{2} \vert \}
\end{equation}
\end{propos}
\begin{proof}
Let us suppose that $ c = \vert c_{2} \vert \geq \vert c_{1} \vert$, so that we can fix $\phi = 0$ in Eq.(\ref{eq:CMmeas}). Then, for $\mu_{s} \to 0$, one can explicitly compute $\lambda_{m} = a - c^{2}/b$. But the initial state $\rrho_{AB}$ is WNS 
if and only if the least classical conditional state \emph{is nonclassical}, which amounts to $\lambda_{m} < 1/2$, as stated by Eq.(\ref{eq:WNScanon}). Otherwise, 
if $ c = \vert c_{1} \vert > \vert c_{2} \vert$, one should choose $\phi = \pi$ to arrive at the same conclusion.
\end{proof}
\noindent
We call this property \emph{weak} nonclassical steering because it does not imply entanglement. Indeed, there are (non isolated) choices for the values of $a,b,c_{1} , c_{2}$ that correspond to physical states ($\ssigma > 0$ and fulfilling UR) that 
are separable and WNS, e.g. $a=b=13.9$, $c_{1} = 4.6$, $c_{2} = -13.7$. Besides, 
there exist WNS states with $ c_{1} c_{2} > 0$, which is a sufficient condition for separability. Motivated by these results, we introduce 
the following more stringent notion of nonclassical steering:
\begin{defn}
A two-mode Gaussian state $\rrho_{AB}$ in canonical form is called \emph{strongly nonclassically steerable} (SNS) ($B \to A$) if the measurement of 
\emph{any} quadrature on mode $B$ generates a nonclassical conditional state
of mode $A$. 
\end{defn}
\noindent
Following the proof of Proposition\ref{propos:WNScanon}, we immediately conclude:
\begin{propos}
\label{propos:SNScanon}
A two-mode Gaussian state $\rrho_{AB}$ in canonical form is SNS ($B \to A$) if and only if the parameters of its CM satisfy:
\begin{equation}
\label{eq:SNScanon}
    a - c'^{2}/b \ < 1/2\,, \quad {c'} = 
    \min \{ \vert c_{1} \vert, \vert c_{2} \vert \}
\end{equation}
\end{propos}
\begin{proof}
The least nonclassical conditional state is reached, among all quadrature measurements, by the ``wrong'' choice of phase ($\phi = \pi$ for $\vert c_{2} 
\vert \geq \vert c_{1} \vert$ and $\phi = 0$ otherwise). Therefore, it is 
sufficient to demand that the minimum eigenvalue of $\ssigma_{c}$ is less
than $\frac{1}{2}$ also in this case, thereby arriving at Ineq.(\ref{eq:SNScanon}).
\end{proof}
In order to generalize these definitions from two-mode Gaussian states in canonical form to \emph{all} Gaussian states of two modes, we should take into account 
(local) single-mode squeezing transformations, which may alter the 
nonclassicality of each mode independently of their quantum correlations. However, since any two-mode Gaussian state can be brought to its \emph{unique} canonical form through LGUTs without altering the correlations, we can extend the definitions in 
the following way:
\begin{defn}
A generic two-mode Gaussian state $\rrho_{AB}$ is called weakly (strongly) nonclassically steerable if the \emph{unique} Gaussian state ${\rrho'}_{AB}$ 
\emph{in canonical form} related to $\rrho_{AB}$ by LGUTs is weakly (strongly) nonclassically steerable. 
\end{defn}
In order to extend also the results regarding the necessary and sufficient conditions for WNS/SNS, we need to specify the effect of LGUTs on $\ssigma_{c}$. Any Gaussian unitary transformation is implemented by a symplectic linear transformation in the phase space formalism, and viceversa. Therefore a LGUT on a two-mode system is described by an element $S_{A} \oplus S_{B}$ acting on quantum phase space, where $S_{A(B)} \in \SL_{A(B)}(2)$. The $2 \times 2$ blocks of a generic $\ssigma$ 
are transformed according to:
\begin{equation}
    \mathbf{A'} \ = \ S_{A} \mathbf{A} S_{A}^{T} \ \ \ \ \ \mathbf{B'} \ = \ S_{B} \mathbf{A} S_{B}^{T} \ \ \ \ \ \mathbf{C'} \ = \ S_{A} \mathbf{C} S_{B}^{T}
\end{equation}
Let us now suppose that $S_{A} \oplus S_{B}$ brings the initial $\ssigma$ in canonical form, so that $\mathbf{A'} = a' \Id_{2}$, $\mathbf{B'} = b' \Id_{2}$ and $\mathbf{C'} = \diag ( {c'}_{1}, {c'}_{2} )$. The conditional CM $\ssigma_{c}$ resulting from a Gaussian measurement with CM $\ssigma_{M}$ on the initial state with CM $\ssigma$ can be rearranged as:
\begin{equation}
    \ssigma_{c}  \ = \ S_{A}^{T} \left[ \mathbf{A'} - \mathbf{C'} \left( \mathbf{B'} + {\ssigma'}_{M} \right)^{-1} \mathbf{C'}^{T}  \right]  S_{A}
\end{equation}
where the CM of the measurement has been redefined according to ${\ssigma'}_{M} = S_{B}^{T} \ssigma_{M} S_{B}$. We see that performing the measurement (associated with) $\ssigma_{M}$ on the two-mode state with CM $\ssigma$ is equivalent to perform the modified measurement ${\ssigma'}_{M}$ on the canonical form state related to $\ssigma$ and then performing the transformation induced by $S_{A}$ on the resulting conditional CM. This means that we can simply factor out the action of $S_{A}$ because it doesn't interfere with the steering process. Meanwhile, as long as $S_{B}$ does not introduce infinite squeezing, we can still approach the desired limit of ${\ssigma'}_{M}$, acting on the state in canonical form, by taking a limit of $\ssigma_{M}$ with a suitable phase. Finally, to get the necessary and sufficient conditions for WNS and SNS in the general case, we can now rewrite Ineq.(\ref{eq:WNScanon}) and Ineq.(\ref{eq:SNScanon}), replacing $a,b,c_{1}, c_{2}$ with their expressions in terms of symplectic invariants \cite{sera04} $I_{1}  = a^{2}$, $ I_{2}  = b^{2}$, $I_{3} = c_{1} c_{2}$, and 
$I_{4} = ( ab - {c_{1}}^{2} ) (ab - {c_{2}}^{2} )$, 
which are indeed invariant  under LGUTs. 
\begin{propos}
A generic two-mode Gaussian state $\rrho_{AB}$ is WNS (SNS) from mode $B \to A$ 
if and only if its symplectic invariants satisfy the inequality:
\begin{equation}
\label{eq:WNSgeneral}
    \dfrac{ I_{1} I_{2} - {I_{3}}^{2} + {I_{4}} \mp \sqrt{ \left( I_{1} I_{2} - {I_{3}}^{2} + I_{4} \right)^{2} - 4 I_{1} I_{2} {I_{4}} }}{2 I_{2} \sqrt{I_{1}}}  <  \frac{1}{2}\,. \notag
\end{equation}
\end{propos}
\noindent
Strong nonclassical steering obviously implies weak nonclassical steering, but it also implies entanglement. We will show this implicitly by proving a stronger result:
\begin{thm}
\label{thm:SNS-EPRsteer}
A two-mode Gaussian state $\rrho_{AB}$ that is SNS $B \to A$ is 
also EPR-steerable in the same direction, therefore also entangled.
\end{thm}
\begin{proof}
Following \cite{jone07}, EPR-steerability $B \to A$ of a Gaussian state by Gaussian measurements amounts to the \emph{violation} of the inequality $\ssigma + i/2\,  \boldsymbol{\omega}_{A} \oplus \boldsymbol{\mathbb{0}}_{B} \geq  0$ by  
its CM. Exploiting LGUT-invariance, we can restrict the comparison between EPR-steerability and SNS to Gaussian states in canonical form. In this case, keeping 
in mind that $a > \frac{1}{2}$, violation of the above inequality reduces to \cite{jone07,kogi15} $(a - c_{1}^{2}/b) (a - c_{2}^{2}/b) < 1/4$,
which is certainly true under the SNS Ineq.(\ref{eq:SNScanon}). 
\end{proof}
At this point, a question may arise on whether WNS is related to the presence of 
Gaussian Quantum Discord (GQD)  \cite{zurek01,ABC16,paris10,ades10,HGR15}. In particular, one 
may ask whether there is a strictly positive lower bound to GQD for states exhibiting WNS, since Gaussian states with zero GQD, being  factorized, are 
obviously \emph{not} WNS.  By construction of explicit examples, 
we now show that this is not the case. It suffices to consider Gaussian states in canonical form with $a = (n+2)/(2n+1)$, $b=n$, $c_{1} = (2n)^{-1/2}$ and $c_{2} = - [ 2n/ (2n+1)]^{1/2}$, for any integer $n > 2$. By direct computation one shows that the  
CMs are $\geq 0$ and obeying the UR. They are also WNS because they fulfill Ineq.(\ref{eq:WNScanon}). However, their GQDs $\mathcal{D}_{A \vert B}$ and $\mathcal{D}_{B \vert A}$ may 
attain arbitrarily small values in the limit $n \to +\infty$.  
\par
Let us now focus on the relevant class of two-mode squeezed thermal states 
(TMST). The parameters of their CMs are given by ($r \in \reals^+$):
\begin{align}
\label{eq:TMSTCM}
a/b & = \frac12 (1 + N_{A} + N_{B})\, \cosh 2r \pm \ \frac12 (N_{A} - N_{B}) 
\notag \\
c  & =  c_{1}  =  - c_{2}  =  \frac12 (1 + N_{A} + N_{B})\,\sinh 2r 
\end{align}
where $N_{i}$ ( $i=A,B$) denotes the average number of thermal photons in each mode. 
Since TMST are all and only those states whose CM is in canonical form with the additional constraint that $c_{1} = - c_{2} = c$, evidently the conditions for 
WNS and SNS coincide for them: the most nonclassical conditional state on mode 
$A$ is obtained by \emph{any} quadrature measurement on mode $B$. From the proof of Theorem\ref{thm:SNS-EPRsteer}, it is also clear that TMST states are EPR-steerable from one mode to the other \emph{if and only if} they are nonclassically steerable (strongly and therefore also weakly) in the same direction. This observation provides a new, somehow surprising, role for the notion of P-nonclassicality: it is \emph{the} 
property that Alice should check, after Bob's measurement on his mode, to certify that the shared TMST state is indeed entangled; we note that this fact could find applications in one-sided device-independent quantum key distribution \cite{bran12}. Note that the \emph{universal} 
steerability condition for TMST states becomes $\cosh 2r >  1 \ +  2 N_{A} (1 + 2 N_{B} )/(1 + N_{A} + N_{B})$, 
which is readily interpreted as a lower bound on the two-mode squeezing 
needed to make the TMST steerable $B \to A$.
\par
In order to illustrate nonclassical steering for TMST states, we employ plots of 
{\em triangoloids}. Consider the conditional CM of mode $A$ parametrized by 
$(\mu_{c}, \mu_{sc}, \phi_{c})$ as in Eq.(\ref{eq:CMmeas}). For TMST it is 
possible to compute the functional dependence of these parameters on the 
initial TMST parameters $N_{A}, N_{B}, r$ and the POVM's parameters 
$\mu, \mu_{s}, \phi$ \footnote{see Eq.(\ref{eq:condCMresult}) of Supplemental Material}. 
In particular we found that $\phi_{c} = \phi$, thus the phase may be discarded. 
For a fixed TMST state, we can thus plot the region of achievable 
conditional states in the $(\mu_{c}, \mu_{sc})$-space, as 
obtained by considering all the POVM's parameters $\mu$ and $\mu_{s}$. 
These are the curvilinear triangles (triangoloids) in Fig.\ref{fig:triangsymm01-12}, where we also displayed the nonclassical region (light-brown region), i.e. those parameters corresponding to nonclassical states \footnote{see Eq.(\ref{eq:nonclreg}) of Supplemental Material}. The TMST state associated with a 
given triangoloid is nonclassically steerable $B \to A$ when the 
triangoloid intersects the nonclassical region,
as in the right panel of Fig.\ref{fig:triangsymm01-12}.
We shaded the intersection area according to the nonclassical depths, with lighter regions for higher $\mathfrak{t}$. As it may be also appreciated graphically, the decisive point for nonclassical steering of a TMST is the blue, lower vertex of the triangoloid, attained by quadrature detection on mode B: if this point is outside the nonclassical region, all other points of the triangoloid are outside too. Notice that the equivalence of EPR steering and nonclassical steering for TMSTs has a neat graphical interpretation: the light-brown nonclassical region is \emph{the largest} region such that a TMST whose triangoloid intersects it is necessarily entangled.
\begin{figure}[h!]
\centering
\includegraphics[width=0.41 \columnwidth]{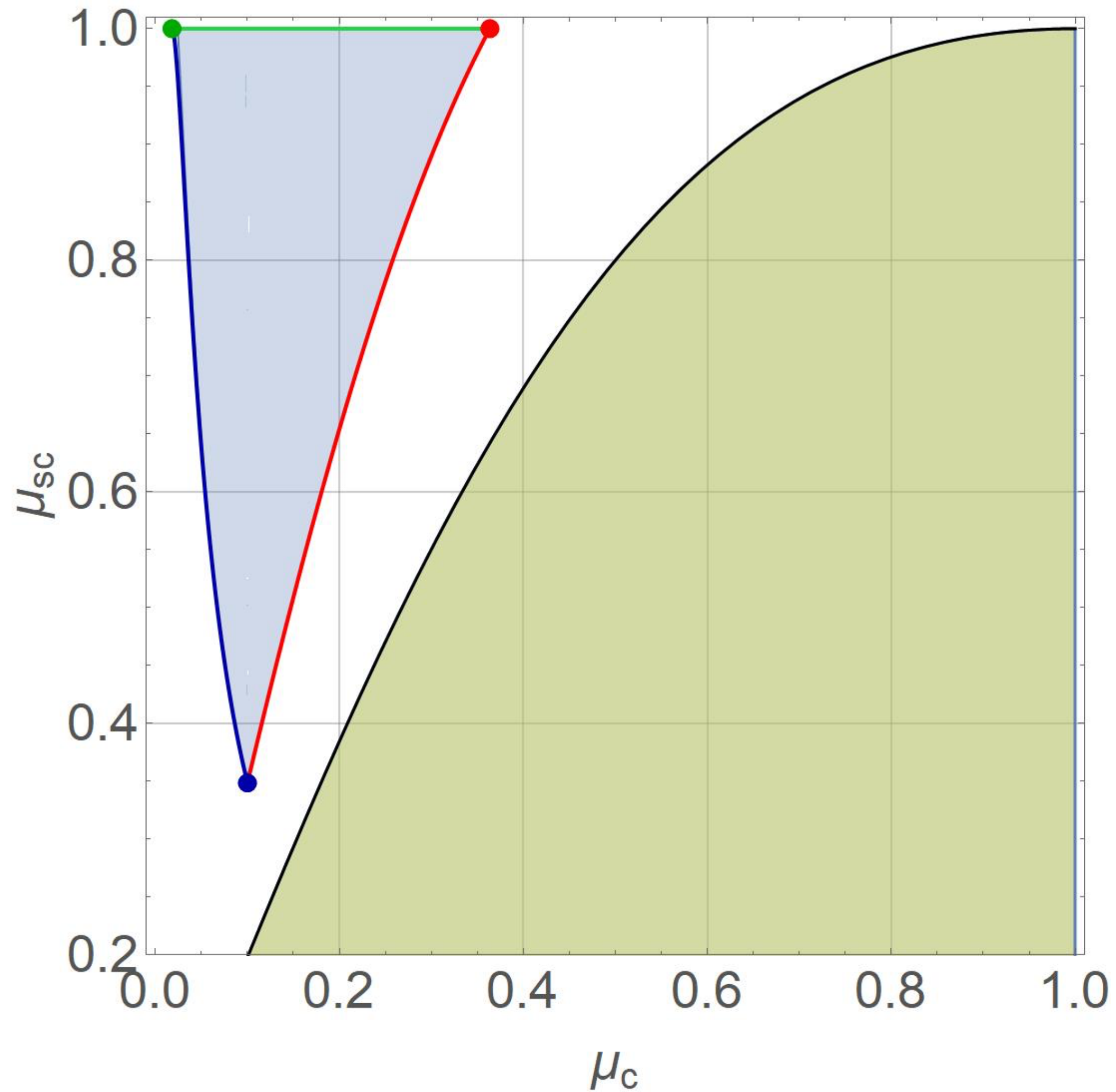} $\quad $
\includegraphics[width=0.41 \columnwidth]{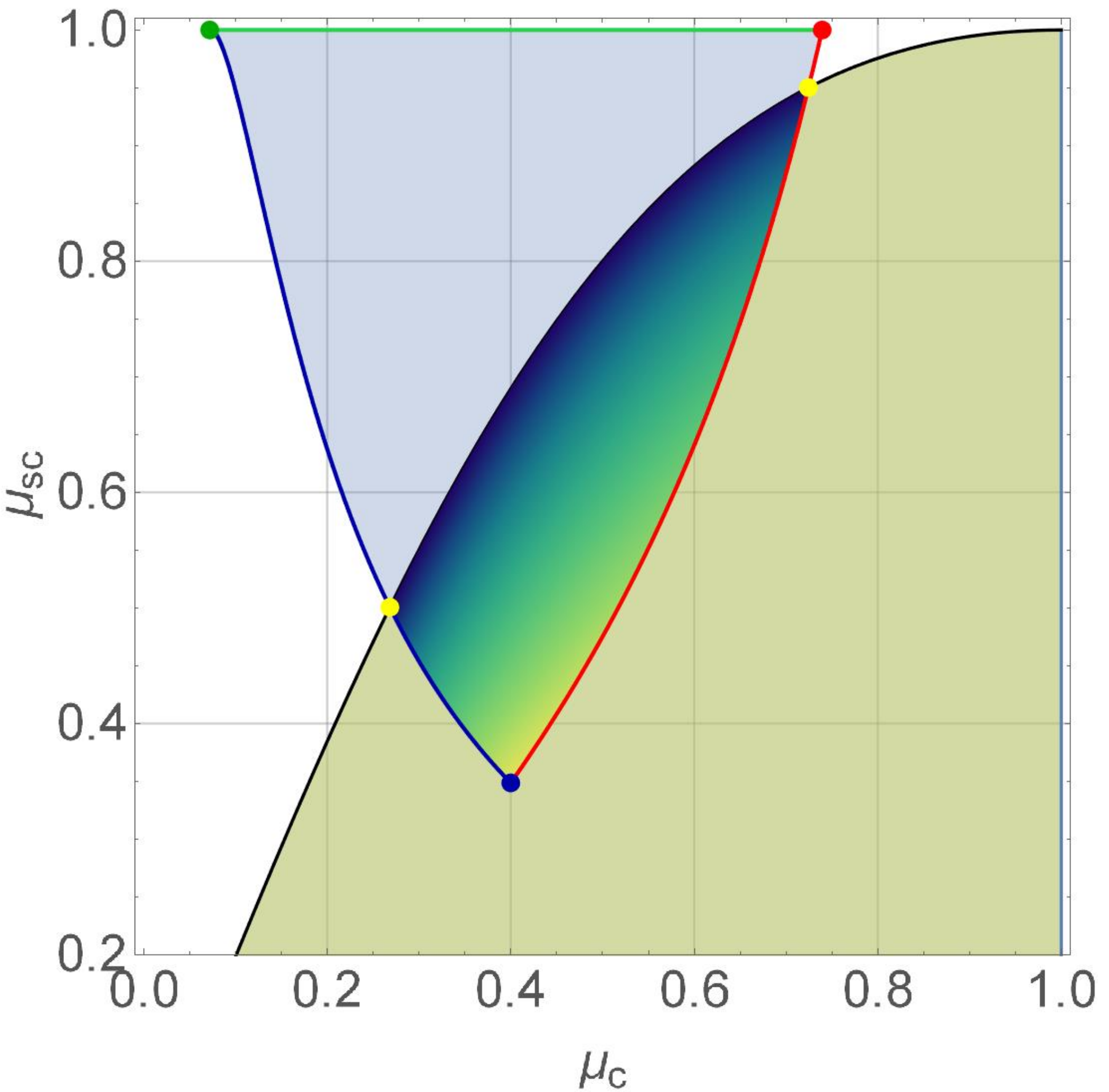}
\caption{\label{fig:triangsymm01-12}
(Left): Triangoloid for TMST state with $N_A = N_B = 4.5$ and $r = 1.2$, 
$\mu_{c}$ is the purity of the conditional state, while $\mu_{sc} = (1 + 2 \sinh^{2} r_{c} )^{-1}$ quantifies squeezing of the conditional state. The light-brown region contains all nonclassical conditional states. (Right):  triangoloid 
for $N_{A} = N_{B} = 0.75$ and $r = 1.2$.}
\end{figure}
\par
The rightmost, red side of the triangoloids is attained by
projective measurements on squeezed displaced vacuum states
(squeezing increases, i.e. $\mu_{s} \to 0$ from the upper red point 
to the lower blue one). Notice also that the uppermost, green side, obtained 
by non-squeezed measurements ($\mu_{s} = 1$) is always at $\mu_{sc} = 1$, i.e. 
the associated conditional states are always classical.
\par
As a final comment, we should mention that the quantities on the left 
sides of (\ref{eq:WNScanon}) and (\ref{eq:SNScanon}) are the conditional 
variances appearing in the Reid EPR-criterion \cite{reid89,reid09}, whose 
test is already experimentally accessible \cite{ou92,mid10}. This is in 
agreement with the well-known result stating that quadrature measurements 
are the best choice for Gaussian EPR steering \cite{kiuk17}. In turn, 
WNS amounts to ask that at least one of such variances is smaller than 
the vacuum value, whereas SNS requires the same to be true 
\emph{for both} these variances separately. EPR-steerability instead 
asks that \emph{the product} of them is smaller than the value attained
by the same quantity on the vacuum \cite{cav09}. This suggests a new 
hierarchy of steering concepts in the Gaussian landscape, 
with WNS being the weakest type, weaker than entanglement, and SNS 
the strongest, while EPR steering is in between, stronger than entanglement 
but weaker than SNS.

\bibliography{apssamp}

\widetext
\clearpage
\pagestyle{empty}
\begin{center}
\textbf{\Large Supplemental Material: Nonclassical steering with two-mode Gaussian states}
\end{center}
\setcounter{equation}{0}
\setcounter{figure}{0}
\setcounter{table}{0}
\setcounter{page}{1}
\makeatletter
\renewcommand{\theequation}{S\arabic{equation}}
\renewcommand{\thefigure}{S\arabic{figure}}
\renewcommand{\bibnumfmt}[1]{[S#1]}
\renewcommand{\citenumfont}[1]{S#1}

\section{1. Derivation of analytical expressions for triangoloid plots}
The covariance matrix (CM) $\ssigma_{c}$ of the conditional state $\rrho_{c, \alpha}$, being single-mode, can be written in the following form:
\begin{equation}
\label{eq:condCM}
\ssigma_{c} \ \ \ = \ \ \ \frac{1}{2 \mu_{c} \mu_{sc}} \left(   \begin{array}{cc} 1 + \kappa_{sc}  \cos \phi & -\kappa_{sc}  \sin \phi \\ - \kappa_{sc}  \sin \phi &  1 - \kappa_{sc}  \cos \phi \end{array}   \right) 
\end{equation}
where $\mu_{c}$ is the purity of the conditional state, $\mu_{sc} = (1 + 2 \sinh^{2} r_{c} )^{-1}$ quantifies the amount of single-mode squeezing $r_{c}$, while $\phi_{c}$ is the squeezing phase and finally $\kappa_{sc} = \sqrt{ 1 - { \mu_{sc}}^{2}}$ for brevity. 
The eigenvalues of $\ssigma_{c}$ are $\lambda_{\pm} = \frac{ 1 \pm \kappa_{sc} }{2 \mu_{c} \mu_{sc}}$ so that, in particular, the conditional state is nonclassical if and only if:
\begin{equation}
\label{eq:nonclreg}
    \lambda_{-} \ = \ \frac{ 1 - \kappa_{sc} }{2 \mu_{c} \mu_{sc}} \ < \ \frac{1}{2} \ \ \ \ \ \implies \ \ \ \ \ {\mu_{sc}} \ < \ \frac{ 2  \mu_{c}}{1 + {\mu_{c}}^{2}}
\end{equation}

which defines implicitly the nonclassical region.
Note that the nonclassicality of $\rrho_{c, \alpha}$ does not depend on $\phi_{c}$ and we can focus just on $\mu_{c}$ and $\mu_{sc}$, which can be retrieved from Eq.(\ref{eq:condCM}) using the following relations:
\begin{equation}
\label{eq:condCMparamextr}
    \det \left[ \ssigma_{c} \right] \ \ = \ \ \ ( 2 {\mu_{c}} )^{-2}  \  , \ \ \ \ \ \ \ \ \ \ \ \Tr \left[ \ssigma_{c} \right] \ \ = \ \ (\mu_{c} \mu_{sc} )^{-1}
\end{equation}
According to the Schur complement formula, the conditional CM $\ssigma_{c}$ for a Gaussian measurement described by $\mu, \mu_{s}$ on mode B of a TMST state with parameters $N_{A}, N_{B}, r$, is given by:
\begin{equation}
\label{eq:condCMform}
    \ssigma_{c} \ \ = \ \ a \cdot \Id_{2} \ - \ c^{2} \left[  \sigma_1 \cdot \left( b \cdot \Id_{2} + \ssigma_{M} \right)^{-1} \cdot \sigma_{1} \right]
\end{equation}
where $\sigma_{1} = \diag(1,-1)$, $\ssigma_{M}$ is the measurement's CM according to Eq.(\ref{eq:CMmeas}) and $a$, $b$ and $c$ are the parameters of the TMST state's CM, defined in Eq.(\ref{eq:TMSTCM}) of the main text. We now define two new parameters to simplify the calculations:
\begin{equation}
    \alpha \ \ := \ \ b + \frac{1}{2 \mu \mu_{s}} \ , \ \ \ \ \ \ \ \ \ \ \ \beta \ \ := \ \ \frac{ \kappa_{s} }{2 \mu \mu_{s}}
\end{equation}
with $\kappa_{s} = \sqrt{1 - {\mu_{s}}^{2}}$ as in Eq.(\ref{eq:CMmeas}). Noting that $\alpha > \beta \geq 0 $, we may write:
\[    \left( b \cdot \Id_{2} + \ssigma_{M} \right)^{-1} \ \ = \ \ \frac{1}{ \alpha^{2} - \beta^{2}} \left(  \begin{array}{cc} \alpha + \beta \cos \phi & - \beta \sin \phi \\
- \beta \sin \phi & \alpha - \beta \cos \phi \end{array} \right)      \]
which can be inserted in Eq.(\ref{eq:condCMform}) to arrive at:
\begin{equation}
    \ssigma_{c} \ \ = \ \ a \cdot \Id_{2} - \frac{c^2 }{ \alpha^{2} - \beta^{2}} \left(  \begin{array}{cc} \alpha - \beta \cos \phi & - \beta \sin \phi \\
- \beta \sin \phi & \alpha + \beta \cos \phi \end{array} \right)
\end{equation}
At this point, $\phi$ is still the phase of the measurement. However, we can now apply Eq.(\ref{eq:condCMparamextr}) and solve for $\mu_{c}$ and $\mu_{sc}$ to get the final result:
\begin{equation}
\label{eq:condCMresult}
    \boxed{ \ \ \ \ \ \ 
    \begin{aligned}
     & \mu_{c}  \ \ = \ \ \frac{1}{2} \sqrt{ \frac{ \alpha^2 - \beta^2}{ (c^2 - a \alpha)^2 - a^2 \beta^2} } \\
     & \mu_{sc}  \ \ = \ \  \frac{ \sqrt{ (\alpha^2 - \beta^2) \left[  (c^2 - a \alpha)^2 - a^2 \beta^2    \right] }   }{a( \alpha^2 -\beta^2) - \alpha c^2} 
    \end{aligned} \\[2pt]
    \ \ \ \ \ \ }
\end{equation}
and we see that $\mu_{c}$ and $\mu_{sc}$ are independent of $\phi$, so it must be that $\phi_{c} = \phi$ and the phase becomes irrelevant for the conditional nonclassicality, hence for the whole (nonclassical) steering process with TMST states.

\end{document}